\documentclass[11pt]{amsart}

\usepackage{amsmath,amsthm,amssymb}
\usepackage[margin=1in]{geometry}
\usepackage{hyperref}
\usepackage{xcolor}
\usepackage{enumerate}
\usepackage{graphicx}
\usepackage{booktabs}
\usepackage{fancyhdr}
\usepackage[colorinlistoftodos,prependcaption,textsize=tiny]{todonotes}
\usepackage{lineno}
\usepackage[foot]{amsaddr}

\newtheorem{theorem}{Theorem}
\newtheorem{lemma}[theorem]{Lemma}

\theoremstyle{definition}
\newtheorem{definition}[theorem]{Definition}
\newtheorem{remark}[theorem]{Remark}
\newtheorem{example}[theorem]{Example}
\newtheorem{fact}[theorem]{Fact}
\newtheorem{claim}[theorem]{Claim}

\newcommand{\E}{\mathbb{E}}
\renewcommand{\Pr}{\operatorname{Pr}}
\newcommand{\R}{\mathbb{R}}
\newcommand{\eps}{\varepsilon}
\newcommand{\norm}[1]{\left\lVert #1 \right\rVert}
\newcommand{\abs}[1]{\left\lvert #1 \right\rvert}

\title{{On GPT Astra's Proof of Lower Bound on Differentially Private Continual Counting}}
\author{Jalaj Upadhyay}
\address{Rutgers University}
\email{jalaj.upadhyay@rutgers.edu}
\date{\today}

\begin{document}
\maketitle
\begin{abstract}
    The goal of this note is to give a detailed proof, to the best of our understanding, of the recent presentation by Harrison and Leeman (arXiv:2609.17650v01 and arXiv:2609.17650v02) of the proof by Astra on the lower bound for differentially private continual counting. We believe a more natural and easy proof is possible and hope that this note will help in that effort. 
    
    Prior to the initial preprint by Harrison and Leeman (arXiv:2609.17650v01), Bairaktari and Larsen (arXiv:2607.00876) gave an elegant proof to show a lower bound of $\Omega(\log^{3/2}(n))$ for both pure and approximate-DP continual counting, and in personal communication had informed us that they have a proof of optimal  $\Omega(\log^{2}(n))$ for pure-differential private continual counting as well. They have subsequently published their $\Omega(\log^{2}(n))$ bound, which is now a joint work of Bairaktari, Dahl, and Larsen (arXiv:2607.00876v3). Their new result is an elegant extension of their technique for approximate-differential privacy. Although the two proofs are technically different, the Astra argument uses related tree geometry introduced in Bairaktari and Larsen.  

    \bigskip \bigskip  
    \noindent \textbf{Disclaimer.} We do not claim any intellectual merit in the proof. The note started as an attempt to better understand and work out the details of the preprint posted by Harrison and Leeman, and to disseminate it to the level accessible to a wider audience. For this reason, we have tried to make the proof self-contained and rigorous, proved many central claims that were claimed but not proved in HL26. That means that we have also used fully detailed and mostly elementary proofs of three claims in HL26 verbatim with proper accredition and explicitly mentioning them. 
\end{abstract}








\clearpage
\section{Introduction}
In a recent preprint, Harrison and Leeman~\cite{harrison2026tight} (which we call HL26 in the rest of this note) published a preprint of a proof by ChatGPT 6 Astra which resolved one of the most interesting problems in differential privacy of the error bound on the counting in the continual observation model. This note originated as a byproduct of our attempt to try to understand this result and to disseminate the result to a level accessible to a wider audience. Therefore, we have tried to keep all the proofs as rigorous as possible, proved results that are claimed and not proved in the original preprint, and verbatim included proofs of three claims from HL26 with proper attribution. For the ease of the readers, we have included remarks at various points in the note that express the best of our understanding.  

\subsection*{Notations.} We follow the notation used in Pillutla et al.~\cite{pillutla2025correlated} for vectors and matrices, and $(a)_+ = \max\{a,0\}$ for $a \in \mathbb{R}$. Other than that, we use the notation used in the preprint~\cite{harrison2026tight} to ensure easy transition for a reader (see Table~\ref{tab:notations}). 

\medskip
\noindent
\textbf{Problem Statement.} For $x\in\{0,1\}^n$, our goal is to study mechanism $\mathcal M:\{0,1\}^n\to\R^n$ that outputs an approximation
to the prefix-sum vector $A_nx$, where
\[
(A_nx)[t]:=\sum_{i=1}^t x[i], \qquad t\in[n].
\]

Following the standards in the literature, we use the following error metric~\cite{pillutla2025correlated}. 
\begin{align*}
    \alpha_\infty(\mathcal M) &:=\max_{x\in\{0,1\}^n} \E_{\mathcal M}\!\left[\norm{\mathcal M(x)-A_nx}_\infty\right],\\ 
    \mathrm{MeanSE}(\mathcal M) &:=\max_{x\in\{0,1\}^n} \E_{\mathcal M}\!\left[\frac1n\norm{\mathcal M(x)-A_nx}_2^2\right],\\
    \mathrm{MaxSE}(\mathcal M) &:=\max_{x\in\{0,1\}^n} \max_{1 \leq t \leq n}\E_{\mathcal M}\!\left[ \left({\mathcal M(x)[t] - (A_nx)[t]}\right)^2\right].
\end{align*}

It is easy to see that $\mathrm{MaxSE}(\mathcal M)\ge \mathrm{MeanSE}(\mathcal M)$ and $\mathrm{MaxSE}(\mathcal M)\ge \alpha_\infty(\mathcal M)^2.$

\section{Proof overview and the roles of the main objects}
\label{ssec:proof-overview}
The preprint posted by Harrison and Leeman~\cite{harrison2026tight} (which we call HL26 in the rest of this note) of the Astra proof uses four related objects: \textit{hard} instance defined by a hidden location, a \emph{decoder}, a pair of \emph{probe--potential functions}, and an exponential \emph{score}.  We first explain their roles at a high level along with a very high level overview of the proof, and delegate their formal definitions in the subsequent subsections.

In differential privacy, the basic goal is to ensure that given two ``neighboring" stream and prefix sum computed over one of them, the adversary cannot tell which stream was used for generating the prefix sum. In the context of continual counting, two streams are neighboring if they differ in one time epoch. So, in order to prove a lower bound, we want to roughly say that, if the algorithm is highly accurate, then an adversary can find at which time epoch the neighboring stream are different. In other word, central to any lower bound proof is the construction of a hard instance. 

\medskip 
\noindent
\textbf{Hard Instance.} 
The Astra proof as presented in the preprint of Harrison and Leeman~\cite{harrison2026tight} construct the hard instance as follows:
\begin{enumerate}
    \item we sample a binary vector $X\in\{0,1\}^m$ and an independent index $J\in[m]$, both uniformly at random\footnote{Using random instance as a hard instance was an idea first used in Andersson et al.~\cite{andersson2024count}.}.  The proof then flips the $J$-th bit and obtains $X' := X\oplus e_J,$ where $e_J$ is the $J$-th standard basis vector. 
    \item Each bit is repeated $k$ times using an expansion operator $E_k:\{0,1\}^m \to \{0,1\}^n$, so the expanded streams $E_k(X)$ and $E_k(X')$ corresponding to $X$ and $X'$ differ exactly on the $J$-th coordinate block.
\end{enumerate}

The goal of the adversary to find the index $J$. 
\begin{remark}
    [Hard Instance and Its Relation to Previous Works]
    Andersson et al.~\cite{andersson2024count} were the first to use random bit stream in their first lower bound on approximate-DP under MaxSE error metric and using the expansion operator along with group privacy was an idea used in Dwork et al.~\cite{dwork2010differentially} for the first lower bound in pure-DP. In some sense, the hard instance in this paper uses the idea from both of these works. The proof of BL26 randomizes a leaf index for an averaging argument, but does not use a random leaf in the manner similar to this proof. To the best of our understanding, the main key innovative ideas is to pick the \textit{hidden location} also uniformly at random.
\end{remark}



\medskip
\noindent 
\textbf{Decoder and Probe Function.} 
Unlike Bairaktari and Larsen \cite{bairaktari2026binary} (which we call BL26 in the rest of this note), Astra works with $4$-ary fully balanced tree, so instead of binary choice as in BL26, Astra has to identify for every internal node $u$ of the tree, which of the edges are more ``useful".  For this, they define a \textit{decoder} algorithm and a \textit{probe} function. The decoder labels two of the four outgoing edges as \emph{favored}: one edge from the even-indexed child and one edge from the odd-indexed child.  This labeling process is a local decision in the sense that it does not depend on any other node, and its goal is to find which child of each parity is more likely to contain the hidden leaf $J$.  The decision is made by comparing short summaries of the residual on selected child intervals,   which are called the \emph{probes}. We let the notation $\Phi$ define the class of probe functions and by $\phi$ a particular instantiation of the probe function.

\medskip
\noindent 
\textbf{Potential Function.}
The local decisions does not itself give us a lower bound on the mechanism's  error mainly because all three error metrics depend on every coordinate (leaves in our case), while decoder and probe function only make local decisions. We therefore pair each probe with a corresponding nonnegative \emph{potential function}.  At a high level, the potential measures the ``spread" of the error on a tree interval. The (probe, potential) pair is designed in a manner so that a wrong local decision can be charged to a decrease of the potential from a parent interval to its children.  Just like in optimization literature, in this case as well, the decrease in potential function telescope over the tree.  Therefore, the expected number of edges that are not favored on the root-to-$J$ path is bounded by the potential of the root, which in turn is bounded by the error of the mechanism.

\medskip
\noindent 
\textbf{Score Function.}
Finally, the proof introduces a hypothetical object called the \textit{score function}, which, at a high level, records partial success at every level of the $4$-ary tree. 
For a parameter $\tau>0$, the proof defines a \textit{score} function, 
\[
S : \Phi \times \{1,\cdots, m\} \times \mathbb R \times \mathbb R \to \mathbb R_+
\]

The score function is chosen in a manner such that its expected value lower bounds the privacy loss on the expanded stream and upper bounds a certain function of the achievable accuracy. Comparing the two gets us the desired lower bound. 

Jumping ahead, in this paper, we pick the score function to be the {\em exponential function} for both the probe functionsr, and the last variable is picked to be 
\begin{align}
    \label{eq:tau-defn}
    \tau:=\Theta \left( \min \left\{ 1, \frac{\log(\varepsilon/\delta)}{\log(n\varepsilon)} \right\}  \right) 
\end{align}
for some absolute constant $C>0$. Here, $\varepsilon,\delta$ are the privacy parameter. So, for the ease of presentation and to be consistent with the notation in HL26, we use the notation $S^\phi(\cdot, \cdot)$ for the score function. 


\subsection*{High Level Proof Overview.} 
Now that we have all the ingredients, we give a high level proof overview. Let $W:=\mathcal M(E_k(X))\in\R^n$ 
be the mechanism output on the original expanded stream.  We call
\begin{align}
   \label{eq:normalized-error}
   Z:=\frac{(W-A_nE_k(X))_{[q]}}{k}\in\R^q
\end{align}
the \emph{normalized error} or \emph{unshifted difference}. 
Here, consistent with HL26, we use $v_{[q]}:=(v[1],\ldots,v[q])$ to denote the restriction of vector $v$ to the set $\{1, \cdots, q\}$. To be consistent with HL26, we call $q$ the \textit{active prefix}\footnote{This is implicit in HL26.} (which for readers familiar to continual counting means \textit{current time epoch}).

Let 
\[
R:=\norm{W-A_nE_k(X)}_\infty.
\]
denote the error on the expanded input. 
Then $\|Z\|_\infty\le R/k$. Note that it is not necessarily equal because $Z$ uses only the
first $q$ coordinates, whereas $R$ is defined over all $n$ coordinates.

Analogous to the normalized error, we define the \emph{shifted residual}
\begin{align}
   \label{eq:shifted-residual}
   Y:=\frac{(W-A_nE_k(X'))_{[q]}}{k}\in\R^q.
\end{align}

Let $N^\phi(j,y)$ denote the number of edges not favored by the decoder algorithm on the path from the root node to  the $j$-node. Let $\phi$ be the probe function used by the decoder and let $y$ be the shifted residual.  For a parameter $\tau>0$, we define the \textit{exponential score} as follows:
\begin{align}
    \label{eq:overview-score}
    S^\phi(j,y):=\exp\bigl(-\tau N^\phi(j,y)\bigr).
\end{align}

As mentioned earlier, the proof relies on two lemmas (Lemma~\ref{lem:accuracy} and Lemma~\ref{lem:privacy}) that sandwich the expected value of the score function, followed by an application of Jensen inequality. In particular, in Lemma~\ref{lem:accuracy}, we show that 
\[
\E_{X,J,\mathcal M}[S^{\phi}(J,Y)] \ge \E_{X, \mathcal M}\!\left[\exp\left(-\Theta\left(\tau \cdot \mathsf{err}\over k\right) \right) \right],
\]
where $\mathsf{err}$ is either MaxSE or MeanSE error and $\tau$ is as defined in equation~\eqref{eq:tau-defn}. Now using Jensen's inequality, we lower bound the RHS in the above equation as follows:  
\[
\E_{X, \mathcal M}\!\left[\exp\left(-\Theta\left(\tau \cdot \mathsf{err}\over k\right) \right) \right] \geq \exp\left(- \Theta\left( {\tau \over k} \right) \E_{X, \mathcal M} [\mathsf{err}]\right) 
\]

This gives us the expected error (MaxSE or MeanSE). In total, these two equations gives us
\[
 \exp\left(- \Theta\left( {\tau \over k} \right) \E_{X, \mathcal M} [\mathsf{err}]\right)  \leq \E_{X,J,\mathcal M}[S^{\phi}(J,Y)].
\]

Then in Lemma~\ref{lem:privacy}, we show that 
\[
  \E_{X,J, \mathcal M}[S^\phi(J,Y)]
  \;\le\;
  e^{k\varepsilon}\!
  \left(\frac{1+e^{-\tau}}{2}\right)^{\!H}
  +\,\delta\,\frac{e^{k\varepsilon}-1}{e^\varepsilon-1}.
\]

Combining the last two equations and setting $k=\Theta(\log(n/\varepsilon))$ and $H=\Theta(\log(n))$, we have the claimed lower bound after some arthematic. 

\begin{remark}
    [Importance of Score Function]
    We note that the probe function and the potential function are the most natural function one can imagine for the corresponding error metric. Therefore, the most innovative part of the proof lies in defining the score function. 
\end{remark}

Throughout this note, we stay consistent with the notations used in HL26 even though it means we overload certain notation. This is to make an seamless transition from the current version of HL26 and this note. The following table contains all the notation that we need in our description and used in HL26. 
\begin{table}[ht]
\centering
\renewcommand{\arraystretch}{1.3}
\begin{tabular}{c|l}
\toprule
\textbf{Notation} & \textbf{Meaning} \\
\midrule
$n$ & length of the original stream \\
$H$ & height of the complete $4$-ary tree \\
$m$ & number of leaves and blocks; $m=4^H$ \\
$k$ & number of repeated coordinates in each block \\
$q$ & length of the active prefix; $q=km\le n$ \\
$E_k$ & expansion operator that repeats each bit $k$ times \\
$I_j$ & $j$-th stream-coordinate block, $\{(j-1)k+1,\ldots,jk\}$ \\
$A_n$ & prefix-sum matrix; $(A_nv)[t]=\sum_{i\le t}v[i]$ \\
$X$ & uniform random vector in $\{0,1\}^m$ \\
$J$ & uniform random target leaf in $[m]$, independent of $X$ \\
$X'$ & flipped vector $X\oplus e_J$ \\
$W$ & mechanism output $\mathcal M(E_k(X))$ \\
$Z$ & normalized error on the original stream \\
$Y$ & shifted residual formed using the flipped stream \\
$Y'$ & reference residual from a fresh run on the flipped stream \\
$s$ & sign of the flipped bit; $s=2X[J]-1\in\{-1,+1\}$ \\
$\phi_u$ & probe of a residual on the stream interval associated with node $u$ \\
$P_u$ & potential paired with $\phi_u$ \\
$N^\phi(j,y)$ & number of non-favored edges on the root-to-$j$ path \\
$S^\phi(j,y)$ & exponential score $e^{-\tau N^\phi(j,y)}$ \\
\bottomrule
\end{tabular}
\vspace{5mm}
\caption{Notation used in the hard instance and the scoring argument in the preprint~\cite{harrison2026tight}. We use the same notation in this note to make it easy for the readers.}
\label{tab:notations}
\end{table}

\subsection*{Organization of the note.}
The rest of the note is organized as follows. In Section~\ref{sec:prelims}, we give the basic preliminaries which includes differntial privacy, detailed description of hard instance, probe and potential functions, decoder algorithm, and score function along with some of our understanding of these choices. The proof in HL26 relies on proving an upper bound on the privacy and a lower bound on accuracy. We give the detail self-contained proof of the accuracy bound result in Section~\ref{sec:main}, that for privacy bound in Section~\ref{sec:privacy-upper}, and then combine the results from these two sections to give the final lower bound proofs in Section~\ref{ssec:combining}. We conclude this note by some discussion on various aspects of this paper in Section~\ref{sec:discussion} based on our persoanl communication and with the authors of HL26 after the first preprint of this note was circulated.

\section{Preliminaries}
\label{sec:prelims}
\subsection{Differential privacy}
\label{ssec:group-amplify}
 For vectors $x,x'\in\{0,1\}^n$, we denote their Hamming distance by
\[
 d_H(x,x'):=\abs{\{t\in[n]:x[t]\ne x'[t]\}}.
\]
A pair of streams is \emph{neighboring} if its Hamming distance is one. We are now ready to define differential privacy.

\begin{definition}[Differential privacy]
    A randomized mechanism $\mathcal M:\{0,1\}^n\to\mathcal Y$ is $(\eps,\delta)$-differentially private if, for all neighboring streams $x,x'$ and all measurable $A\subseteq\mathcal Y$,
    \[
    \Pr[\mathcal M(x)\in A]\le e^\eps\Pr[\mathcal M(x')\in A]+\delta.
    \]
    When $\delta=0$, we call that the mechanism satisfies $\varepsilon$-differentially private.
\end{definition}

The block expansion operator $E_k(\cdot)$ defined in equation~\eqref{eq:block-expansion-op} amplifies the hidden change caused by flipping a random bit, but it also increases the privacy cost.  In particular, by flipping coordinate $J$ of $X$, we get changes in $k$ coordinates of the expanded stream. While this plays an important role in proving a lower bound on the accuracy (Lemma~\ref{lem:accuracy}), it makes 
$E_k(X)$ and $E_k(X')$ at a Hamming distance $k$. To cover this difference in the neighboring stream, we recall the bound on \textit{group privacy}. 

\begin{lemma}[Group privacy~\cite{dwork2014algorithmic}]
    If $\mathcal M$ is $(\eps,\delta)$-differentially private and $d_H(x,x')\le k$, then every measurable $f:\mathcal Y\to[0,1]$ satisfies
    \[
    \E[f(\mathcal M(x))] \le e^{k\eps}\E[f(\mathcal M(x'))]+\delta_k,
    \qquad
    \delta_k:=\delta\sum_{i=0}^{k-1}e^{i\eps} =\delta\frac{e^{k\eps}-1}{e^\eps-1}.
    \]
\end{lemma}

\subsection{The hard instance and the shifted residual}
\label{ssec:instance}
The Astra's proof for the lower bounds for all the error metrics uses the   the same hard distribution with the difference in the proof for both the metric  lying in the choice of probe-potential pair defined in Section~\ref{ssec:probes}. To explain their hard instance, we first need to describe the expansion operator. 

\medskip \noindent
\textbf{Block expansion.}
Fix positive integers $H$ and $k$, set $m:=4^H$, and let $q:=km\le n$.  Define
$E_k:\{0,1\}^m\to\{0,1\}^n$ by
\begin{align}
\label{eq:block-expansion-op}    
  E_k(x)[t] :=
  \begin{cases}
     x[\lceil t/k\rceil], & 1\le t\le q,\\
     0, & q<t\le n.
  \end{cases}
\end{align}

For $j\in[m]$, define the $j$-th block $I_j:=\{(j-1)k+1,\ldots,jk\}$ to be the $k$ contiguous indices. Therefore, $E_k(x)$ is constant on each block $I_j$, with value $x[j]$, and the
blocks partition the currently active prefix $\{1, \cdots, q\}$. For $t\in I_j$, the prefix sum of the \textit{expanded stream} is
\begin{align}
  \label{eq:prefix-sum-block}
  (A_nE_k(x))[t]
  =k\sum_{i=1}^{j-1}x[i]
   +\bigl(t-(j-1)k\bigr)x[j].
\end{align}

At the block boundary $t=jk$ for some $j \in \mathbb N$, equation~\eqref{eq:prefix-sum-block} is 
$k\sum_{i=1}^j x[i]$.  Therefore, flipping the bit $x[j]$  changes the prefix sums gradually
within $I_j$ and by a constant amount $k$ at every coordinate after $I_j$. 

\begin{example}
    Let $e_1, \cdots, e_{m}$ be standard basis of $\mathbb R^{m}$ for $m=n/\log(n)$. Then $E_{\log(n)}(e_i)$ are the hard instances considerd in the lower bound proof in Dwork et al.~\cite{dwork2010differentially}.
\end{example}

\medskip \noindent
\textbf{The random input pair.}
Sample
\[
X\sim\mathrm{Uniform}(\{0,1\}^m),
\qquad
J\sim\mathrm{Uniform}([m]),
\]
independently, and set $X':=X\oplus e_J.$ If we do not have the random indices $J$ on which the input changes, this is the same hard instance used in Andersson et al.~\cite{andersson2024count}. 

The hard instance in this proof is the expanded streams $E_k(X)$ and $E_k(X')$. It is easy to see that they are identical outside $I_J$ and differ at every coordinate of $I_J$.  Therefore $d_H(E_k(X),E_k(X'))=k.$ We will be working with the expanded stream, so the hard instance can be seen as using the idea for two previous works~\cite{andersson2024count, dwork2010differentially} that proved $\log(n)$ lower bound on pure and approximate DP, respectively. 

The following fact is immediate from the description of $X'$ and $J$. 
\begin{enumerate}
    \item $X'$ is uniform on $\{0,1\}^m$ and independent of $J$.
    \item The sign $s:=2X[J]-1$ is uniform on $\{-1,+1\}$.
\end{enumerate}

In the high-level overview of the proof, we defined normalized and shifted residual. A natural question is why do we need these two definitions. We next discuss why the proof introduces this concept to the best of our understanding. 

To get some idea, let us try to understand the relationship between the normalized error $Z$ from the shifted residual $Y$ defined in equations~\eqref{eq:normalized-error} and \eqref{eq:shifted-residual}, respectively. Since the target index $J$ is sampled independently of $X$,  $W=\mathcal M(E_k(X))$ does not depend on $J$, and the normalized error vector $Z$ does not contain the index $J$. On the other hand, we can show that the random vector $Y$  satisfies $Y=Z+g_J,$  where 
\[ 
g_J[t] =\frac{s}{k} \min\bigl\{(t-(J-1)k)_+,k\bigr\}, \qquad s=2X[J]-1\in\{-1,+1\},
\] 

Note that $g_J$ is zero before block $I_J$, changes linearly from $0$ to $s$ inside $I_J$, and equals the constant $s$ after $I_J$ (note that the adversary does not know the value of $J$).  Therefore, $Y$ is obtained by adding to the normalized error $Z$, a step-shaped vector $g_J$ whose form is fixed but whose index and whehter it is $+1$ or $-1$ are hidden. 
 
We make this formal through the following lemma and then discuss its importance in the context of the proof.

\begin{lemma}[Residual decomposition]\label{lem:residual-decomposition}
    For $a\in\R$, write $(a)_+:=\max\{a,0\}$.  For every $t\in[q]$,
    \begin{align}
        \label{eq:residual-decomp}
        Y[t]-Z[t] =\frac{s}{k}\min\bigl\{(t-(J-1)k)_+,k\bigr\},
        \qquad \text{where} \qquad s=2X[J]-1.
    \end{align}
\end{lemma}

\begin{proof}
    We first have 
    \begin{align}
        Y[t]-Z[t]
        &=\frac1k\left((A_nE_k(X))[t]-(A_nE_k(X'))[t]\right) =\frac1k\sum_{i=1}^t\left(E_k(X)[i]-E_k(X')[i]\right).
        \label{eq:yz-diff}
    \end{align}

    Since $X$ and $X'$ differ only at $J$, we have $E_k(X)[i]-E_k(X')[i] =s\,\mathbf 1\{i\in I_J\}.$ Substituting it into equation~\eqref{eq:yz-diff}, we get 
    \[
    Y[t]-Z[t] =\frac{s}{k}|I_J\cap[t]|.
    \]
    Finally, it is easy to see that 
    \begin{align}
        \label{eq:case-residual}    
         |I_J\cap[t]| =
         \begin{cases}
            0, & t\le (J-1)k,\\
            t-(J-1)k, & (J-1)k<t<Jk,\\
            k, & t\ge Jk,
        \end{cases}
    \end{align}
    which is exactly $\min\{(t-(J-1)k)_+,k\}$.
\end{proof}

Now as promised, we return to the question as to why is this lemma important or what it is conveying? The case form in equation~\eqref{eq:case-residual} of the residual decomposition lemma gives three regime of the decomposition. 
\begin{itemize}
    \item \emph{Before the transition block:} if $t\le(J-1)k$, then $Y[t]=Z[t]$.

    \item \emph{Inside the transition block:} if $(J-1)k<t<Jk$, then
    \[
    Y[t]-Z[t]=\frac{s}{k}(t-(J-1)k).
    \]
    We call this affine change from $0$ to $s$ the \emph{ramp}.

    \item \emph{After the transition block:} if $t\ge Jk$, then $Y[t]=Z[t]+s$.  We call the constant offset $s$, whose magnitude is one, 
  the \emph{unit step}.
\end{itemize}

Therefore, as mentioned earlier, $Y$ is the normalized error $Z$ plus a ``\textit{ramp}" on $I_J$ followed by a constant unit step defined by $s \in \{-1,1\}$. Note tgat the entire goal of the decoder is to find this transition point from $Y$.  We will see later in Section~\ref{ssec:cross} that the \textit{cross-parity} rule, that makes a decision whether a child of an internal node is ``useful" in finding the hidden index $J$ or not, ensures that a decision about any edge from the root note to the $J$-th leaf never probes the interval of the child that contains $I_J$. As a result, the local comparison only sees the regions on which the inserted
signal is constant, and not the ramp. 

\subsection{The $4$-ary tree}
\label{ssec:cross}
BL26 uses the more natural approach of binary tree basis. However, unlike BL26,  the proof of Astra uses a $4$-ary tree. In the context of binary input, this  seems very unnatural construction, and  unfortunately, the current version does not provide any explanation or intuition. We next try to understand the edge-labeling rule used by the Astra decoder and develop some intuition as to why $4$-ary is used in the proof. We return to this question again at the end of this section.

\medskip \noindent
\textbf{Tree structure.}
Let $\mathcal{T}$ be the complete ordered $4$-ary tree of height $H$ with leaf set $[m]$, where $m=4^H$.  We identify each node with the contiguous interval of leaf indices below it.  In other words, the root is $[1,m]$, and a node at depth $d$ contains $4^{H-d}$ leaves.  If an internal node is
\[
u=[a,b], \qquad L:=|u|=b-a+1,
\]
then $L$ is divisible by $4$, and the four children of $u$, listed from left
to right, are
\begin{align*}
    u_0&=[a,\,a+L/4-1],\\
    u_1&=[a+L/4,\,a+L/2-1],\\
    u_2&=[a+L/2,\,a+3L/4-1],\\
    u_3&=[a+3L/4,\,b].
\end{align*}

We call $u_0,u_2$ the \emph{even children} and $u_1,u_3$ the \emph{odd children}.  For any node $v=[a,b]$, recall that
\[
I_v:=\bigcup_{j=a}^b I_j =\{(a-1)k+1,\ldots,bk\} \subseteq \{1, \cdots, q\}
\]
is the associated interval of stream coordinates.  Therefore, if the target leaf $J$ lies in child $u_i$, then the transition block $I_J \subseteq I_{u_i}$ and is disjoint from the other three child intervals.

\subsection{Probe and potential functions}
\label{ssec:probes}
For a tree node $u=[a,b]\subseteq[m]$, let
\[
I_u:=\bigcup_{j=a}^b I_j =\{(a-1)k+1,\ldots,bk\} \subseteq \{1, \cdots, q\}
\]
be the corresponding interval of stream coordinates.  The Astra decoder makes a local decision at $u$ by comparing scalar summaries of the residual on selected child intervals.  We call such a summary a \emph{probe}. Since the probe make local decision and our goal is to understand a global property of the mechanism, we need to convert the information provided by probe to a more global informtation. For this, every probe has an associated \textit{potential function}, which measures the variation of the residual on an interval. In this note, we use the following probe-potential function. In retrospect, these are the most natural pair for the corresponding error metrics. 

\begin{definition}[Probe--potential pairs]\label{def:probe}
    For $v\in\mathbb{R}^d$ and a nonempty set $A\subseteq[d]$, define
    \begin{align*}
        c_A(v) 
        &:=\frac12\left(\max_{i\in A}v[i]+\min_{i\in A}v[i]\right),
        & r_A(v)
        &:=\max_{i\in A}v[i]-\min_{i\in A}v[i],\\[3pt] \mu_A(v)
        &:=\frac1{|A|}\sum_{i\in A}v[i],
        & V_A(v)
        &:=\frac1{|A|}\sum_{i\in A}\bigl(v[i]-\mu_A(v)\bigr)^2.
    \end{align*}

    The pairs $(c,r)$ and $(\mu,V)$ are called the \emph{midrange--range} and \emph{mean--variance} pairs, respectively.  For a tree node $u$, we use the notation 
    \[
    \phi_u(v):=\phi_{I_u}(v), \qquad P_u(v):=P_{I_u}(v),
    \]
    where $(\phi,P)$ is either $(c,r)$ or $(\mu,V)$.
\end{definition}

The midrange--range pair is used for the MaxSE lower bound, while the mean--variance pair is used for the MeanSE lower bound.  Both pairs satisfy certain important properties that allows us to relate the shifted residual (defined in equation~\eqref{eq:shifted-residual}) with the error of any DP mechanism. We next describe these properties:
\begin{fact}
    [Translation equivariant of probe function and Translation invariance of potential function]
    For every constant $a\in\mathbb{R}$,
    \begin{align}
        \label{eq:translation}
        \phi_A(v+a\mathbf{1})=\phi_A(v)+a,
        \qquad
        P_A(v+a\mathbf{1})=P_A(v),
    \end{align}
    where $\mathbf{1}$ is the all-ones vector on $A$.  Thus the probes are
    \textit{translation equivariant} and the potentials are \textit{translation invariant}.
\end{fact}

\medskip \noindent
\textbf{The midrange--range pair.} 
For the MaxSE case, the potential of the root node satisfies
\[
r_{[q]}(Z)\le 2\|Z\|_\infty\le \frac{2R}{k},
\]
where $R:=\|W-A_nE_k(X)\|_\infty$.  

The following claim, stated almost verbatim (apart from changing notations and stylist preference) from HL26 along with its proof, relates change in the midrange to a decrease in the range potential. The proof is elementary and follows from the definition of probe and potential functions.

\begin{claim}
    [{\bf Harrison and Leeman~\cite[Lemma 4]{harrison2026tight}}]
    \label{claim:midrange}
    For $v\in\mathbb{R}^d$ and nonempty sets $A'\subseteq A\subseteq[d]$,
    \[
   |c_{A'}(v)-c_A(v)|
   \le \frac{r_A(v)-r_{A'}(v)}2.
   \]
\end{claim}

\begin{proof}
    The values of $v$ on $A'$ lie in the interval $[c_{A'}-r_{A'}/2,c_{A'}+r_{A'}/2]$, which is contained in $[c_A-r_A/2,c_A+r_A/2]$.  Comparing the left endpoints gives $c_{A'}-c_A\ge-(r_A-r_{A'})/2$, while comparing the right endpoints gives $c_{A'}-c_A\le(r_A-r_{A'})/2$.  The two inequalities imply the claim.
\end{proof}

\begin{remark}
    The midrange is also $1$-Lipschitz with respect to $\ell_\infty$, i.e., $|c_A(v)-c_A(w)|\le\|v-w\|_\infty.$  However, we need Claim~\ref{claim:midrange} here because it  controls a probe difference across nested intervals by a decrease in the potential.    
\end{remark}

\medskip
\noindent
\textbf{The mean--variance pair.}
For the mean-squared-error analysis, the unnormalized squared error on the active prefix is 
\[
Q:=\frac1q\sum_{t\le q} \bigl(W[t]-(A_nE_k(X))[t]\bigr)^2.
\]
Since $Z=\frac1k(W-A_nE_k(X))_{[q]}$,
\[
\frac1q\sum_{t\le q}Z[t]^2=\frac{Q}{k^2}.
\]

Therefore, the most natural probe function is the mean, and the corresponding potential is the  variance. The following lemma, stated almost verbatim (apart from changing notations and stylist preference) in Harrison and Leeman~\cite{harrison2026tight} along with its proof, plays the same role as midrange stability. The proof is elementary and follows from the definition of probe and potential functions.

\begin{lemma}
    [{\bf Harrison and Leeman~\cite[Lemma 5]{harrison2026tight}}]
    \label{lem:totvar}
    Let $A_0,\ldots,A_{B-1}$ partition a nonempty set $A$, and let $w_i:=|A_i|/|A|$.  Then
    \[
    V_A(v)-\sum_{i=0}^{B-1}w_iV_{A_i}(v) =\sum_{i=0}^{B-1}w_i\bigl(\mu_{A_i}(v)-\mu_A(v)\bigr)^2.
    \]
\end{lemma}

\begin{proof}
    For $j\in A_i$,
    \[
    v[j]-\mu_A(v) =\bigl(v[j]-\mu_{A_i}(v)\bigr)  +\bigl(\mu_{A_i}(v)-\mu_A(v)\bigr).
    \]
    After squaring and summing over $j\in A_i$, the cross term vanishes because $\sum_{j\in A_i}(v[j]-\mu_{A_i}(v))=0$.  Summing over $i$ and dividing by $|A|$ proves the identity.
\end{proof}


\begin{remark}
    While the mean--variance pair has no direct counterpart in the result BL26, we believe that the proof of BL26 can be extended to give a mean-squared-error bound as well.  This is because the law of total variance plays the same structural role as the recursive range inequality in BL26.  
\end{remark}

\subsection{Decoder and Cross-parity rule}
\label{ssec:cross}
Recall that the main goal of the decoder is to label two outgoing edges at every internal node as favored  in such a way that an accurate mechanism tends to favor the edge on the root-to-$J$ path. We describe the main idea behind the decoding algorithm. 

Let $\mathcal{T}_{\mathrm{int}}$ denote the internal nodes of $\mathcal{T}$.
Fix one of the probes $\phi\in\{c,\mu\}$ defined in
Section~\ref{ssec:probes} (as we discussed earlier, it does not matter which probe function we use as long as we use the corresponding potential function, so the reader can pick the probe function they are most comfortable with).  For a residual vector $y\in\mathbb{R}^q$, the decoder labels one even and one odd outgoing edge at every $u\in\mathcal{T}_{\mathrm{int}}$.

\begin{definition}
    [Cross-parity rule]
    Define the favored even child\footnote{While unfortunate and not our preference, we overload the letter $e$ to also define the even nodes even though the letter is reserved for the exponential function. This is to be consistent with the original preprint.} and odd child by
    \begin{align}
        \label{eq:favored}
        e_u^\phi(y) :=
        \begin{cases}
            u_2, & \text{if }|\phi_{u_1}(y)-\phi_{u_3}(y)|\ge\frac12,\\
            u_0, & \text{if }|\phi_{u_1}(y)-\phi_{u_3}(y)|<\frac12,
        \end{cases}
        \quad \text{and} \quad 
        o_u^\phi(y) :=
        \begin{cases}
            u_1, & \text{if }|\phi_{u_0}(y)-\phi_{u_2}(y)|\ge\frac12,\\
            u_3, & \text{if }|\phi_{u_0}(y)-\phi_{u_2}(y)|<\frac12.
        \end{cases}
    \end{align}
\end{definition}

Then we can denote the favored-edge set by 
\begin{align}
    \label{eq:favored-edges}    
    \mathcal{E}_{\mathrm{fav}}^\phi(y) :=\bigcup_{u\in\mathcal{T}_{\mathrm{int}}} \bigl\{(u,e_u^\phi(y)),(u,o_u^\phi(y))\bigr\}.
\end{align}

Therefore, two out the four outgoing edges at each internal node are favored. We use the term \emph{cross-parity} because the choice between the even children is made by comparing the odd children, while the choice between the odd children is made by comparing the even children.

\begin{remark}
    [Why cross-parity rule?]
    \label{rem:cross-parity}
    It is unclear and is not discussed (even though we feel it is very crucial in the understanding of the proof) is why the proof uses opposite-parity probe. On its own, using such a method is very unnatural. Therefore, we attempt to give some intuition based on our basic understanding as to why the proof and its usage of $4$-ary tree crucially need such a trick. 

    Recall that $Y=Z+g_J,$ where $g_J$ is zero before $I_J$, varies linearly from $0$ to $s\in\{-1,+1\}$ on $I_J$, and equals $s$ after $I_J$.  The only coordinates on which $g_J$ is not constant are those in $I_J$.  Therefore a probe can be rewritten entirely in terms of $Z$ whenever its coordinate interval is disjoint from $I_J$.

    Now suppose $J\in u_i$.  To decide whether the edge toward $J$ is favored, the cross-parity rule compares the two children of parity opposite to $i$. Neither of these two children is $u_i$.  Therefore, both probed intervals are disjoint from $I_J$ and lie entirely before or entirely after $I_J$.  On each such interval, $g_J$ is a constant, so using translation equivariance of $\phi$, we have 
    \[
    \phi_v(Y)=\phi_v(Z)+\gamma_v, 
    \qquad \gamma_v\in\{0,s\}.
    \]

    In other words, cross-parity prevents ``contamination by the ramp" as the comparison determining whether the true edge is favored never uses a probe on the child interval containing $I_J$.
\end{remark}

\subsection{Why $4$-ary tree?}
\label{ssec:why-4-ary}
Now once we have the understanding of machineries used in the proof, we go back to our initial question: \textit{why use $4$-ary tree and not any other $2^a$-ary tree and in particular binary tree as in BL26?}   The use of $4$-ary tree serves two separate purposes.
\begin{enumerate}
    \item  First, it supports the cross-parity comparison.  The four ordered children split into two parity classes of size two.  The decoder can therefore decide between the two children in one class by probing the two children in the other class.  If the target is in the class being decided, neither probed child contains the transition block.  A binary node has no disjoint opposite-parity pair available for this purpose.

    \item  Second, exactly two outgoing edges are favored at every internal node, independently of the residual.  In the reference experiment, $J$ is uniform and independent of $Y'$.  Conditioned on $Y'=y'$ and on the prefix of the root-to-$J$ path above a node $u$, the next child containing $J$ is uniform in
    $\{u_0,u_1,u_2,u_3\}$. Therefore,
    \[
    \Pr\bigl[(u,\operatorname{child}_J(u)) \notin\mathcal{E}_{\mathrm{fav}}^\phi(y') \mid Y'=y',\text{ previous path choices}\bigr] =\frac12.
    \]
\end{enumerate}

Moreover, the base-$4$ representation of a uniform leaf are independent across levels.  Therefore, the non-favored edges along the path are independent Bernoulli random variables,  $\mathrm{Bernoulli}(1/2)$, and we can compute the MGF, and hence analyze the score function easily. 

We note that it is possible using some clever accounting to even work with higher order ary-tree, but it seems to become very complicated  to analyze. In particular, looking at distribution of $N^\phi(J,Y)$, because of the 4-ary structure and that exactly 2 of the 4 outgoing edges are favored at every node, in the reference experiment, $J$ lands in each of the 4 children with probability $1/4$, so the probability of a non-favored edge at any level is exactly $1/2$, and we get the nice binomial distribution that is very easy to analyze. Such a non-favored probability is also the case in any $2^b$-ary tree (and in particular, for the binary tree). In fact, it seems like using the $4$-ary tree leads to the contamination problem, and which is what makes the proof more intricate to analyze. However, on the other hand, having a $4$-ary tree allows this proof to not have a complex two-level grandchild conditioning argument and a linear-measurement privacy attack as in BL26. These are some of the reasons why we believe that getting a better understanding of cross-parity argument is crucial for understanding this proof.

\subsection{The exponential scoring function}
\label{ssec:score}
The cross-parity decoder labels exactly two outgoing edges at each internal node as favored.  Depending on this labeling, the \textit{exponential scoring function} outputs a value that can be analyzed from both the accuracy and privacy perspectives. For a residual vector $y\in\mathbb{R}^q$, recall the favored-edge set from equation~\eqref{eq:favored-edges}, 
\[
\mathcal{E}_{\mathrm{fav}}^\phi(y)  :=\bigcup_{u\in\mathcal{T}_{\mathrm{int}}} \bigl\{(u,e_u^\phi(y)),(u,o_u^\phi(y))\bigr\}.
\]

For a leaf $j\in[m]$, let $\operatorname{path}(j)$ be the set of the $H$ parent--child edges on the unique path from the root to $j$.  Define $N^\phi(j,y) :=\bigl|\operatorname{path}(j)         \setminus\mathcal{E}_{\mathrm{fav}}^\phi(y)\bigr|.$  That is, $N^\phi(j,y)$ is the number of levels at which the edge toward $j$ is not favored by the decoder.  For a parameter $\tau>0$, define the
\emph{exponential score} to be 
\begin{align*}
     S^\phi(j,y) :=\exp\bigl(-\tau N^\phi(j,y)\bigr) \in [e^{-\tau H},1].
\end{align*}

It equals $1$ when every edge on the root-to-$j$ path is favored and is multiplied by $e^{-\tau}$ for each
non-favored edge.

\section{Accuracy lower bound}
\label{sec:main}
In this section, we build on the concepts and discussion from the previous section to give detailed proof of the lower bound.  Since our goal is to provide exposition and ensure that all the statements are easy to follow, we use the same terminilogy as used in the preprint of Astra's proof. We start with giving a self-contained proof of accuracy lower bound (Lemma 6 in the preprint). 

\begin{remark}
    The original proof suppresses two important arguments and mention them in passing. While the claims in HL26 is correct and these omission do not affects the correctness of their statement, however, we feel that they are essential for a complete rigorous proof and has some subtly that needs to be explicitly spelled out. We enumerate first for the readers who wishes to quickly find them in the current note. Other readers can simply skip this remark. 

    \begin{enumerate}
        \item \emph{Probe decomposition.}  The core computation at each node rewrites $\phi_{u_i}(Y)$ as $\phi_{u_i}(Z)$ or $\phi_{u_i}(Z)+s$ depending on whether $I_{u_i}$ lies before or after $I_J$.  This requires: (a)~the explicit formula $Y=Z+g_J$ from Lemma~\ref{lem:residual-decomposition}; (b)~the fact that the probed child interval $I_{u_i}$ is disjoint from $I_J$, which follows from the cross-parity structure of Section~\ref{ssec:cross}; and (c)~translation equivariance of the probe, which is stated in Definition~\ref{def:probe} but is never cross-referenced in the proof. We use all these steps in Claim~\ref{claim:probe-decomp}.

        \item \emph{Case analysis.}  The bound on the local error probability $p^\phi_u$ (defined in equation~\eqref{eq:p-phi-u-defn}) requires a complete derivation for each location of $J$.  The Astra preprint names the four cases and just state the bound in each of these cases. However, this calculation is more subtle, and their claim do not show why the comparison reduces to $|d_{13}|$ when $J\in u_0$, or why $|d_{02}+s|<1/2$ forces $|d_{02}|>1/2$ when $J\in u_1$. We do it more explicitly in Claim~\ref{claim:p-u-phi}.
    \end{enumerate}

\end{remark}

\begin{lemma}[Accuracy lower bound]
\label{lem:accuracy} 
    Under the distributions of Section~\ref{ssec:instance}, for any $\tau>0$, \begin{align}
        \E_{X,J,\mathcal M}[S^{\mathrm{mid}}(J,Y)]
        &\ge \E_{X, \mathcal M}\!\left[e^{-4\tau R/k}\right],    \label{eq:acc-mid}\\
        \E_{X,J,\mathcal M}[S^{\mathrm{mean}}(J,Y)]
        &\ge \E_{X, \mathcal M}\!\left[e^{-16\tau Q/k^2}\right], \label{eq:acc-mean}
    \end{align}
    where $R:=\|W-A_nE_k(X)\|_\infty$ and $Q:=\frac1q\sum_{t\le q}(W-A_nE_k(X))[t]^2$.
\end{lemma}

\begin{proof}
    Fix a realization of $X$ and $W=\mathcal M(E_k(X))$, and define
    \[ 
    Z:=\frac{(W-A_nE_k(X))_{[q]}}{k}\in\mathbb{R}^q.
    \]
    Note that $Z$ is a function of $X$ and the mechanism's internal randomness, and is independent of $J$.

    We first perform the probe decomposition, which addresses point (1) mentioned above.

    \begin{claim}\label{claim:probe-decomp}
        Let $v=[c,d]$ be a tree node whose stream interval $I_v$ is disjoint from $I_J$.  Then
        \[
        \phi_v(Y) = 
        \begin{cases}
            \phi_v(Z)+s, & \text{if }I_v\text{ lies after }I_J \text{ (i.e., }(c-1)k\ge Jk\text{)},\\
            \phi_v(Z),   & \text{if }I_v\text{ lies before }I_J \text{ (i.e., }dk\le(J-1)k\text{)}.
        \end{cases}
        \]
\end{claim}

\begin{proof}
    By Lemma~\ref{lem:residual-decomposition},
    \[
    Y[t]-Z[t]=g_J[t]:=\frac{s}{k}\min\bigl\{(t-(J-1)k)_+,k\bigr\}.
    \]
    
    For every $t\in I_v$ with $I_v$ after $I_J$, we have $t\ge Jk$, so $g_J[t]=s$. Let $Z|_{I_v}$ denote the restriction of $Z$ to indices in $I_v$. Then  $Y|_{I_v}=Z|_{I_v}+s\cdot\mathbf{1}$, and translation equivariance (equation~\eqref{eq:translation}) gives $\phi_v(Y)=\phi_v(Z)+s$.

    For every $t\in I_v$ with $I_v$ before $I_J$, we have $t\le(J-1)k$, so $g_J[t]=0$, giving $\phi_v(Y)=\phi_v(Z)$.
\end{proof}

We next rigorously performing the missing four case analysis to understand $p_u^\phi$ (point (2) above). 
\begin{claim}
    \label{claim:p-u-phi}
    Let $u$ be an internal node of the tree and suppose $J\in u$.  Define
    \[
    d_{02}:=\phi_{u_2}(Z)-\phi_{u_0}(Z),
    \qquad
    d_{13}:=\phi_{u_3}(Z)-\phi_{u_1}(Z).
    \]
    For a probe function $\phi$, let $p^\phi_u$ be defined as 
    \begin{align}
        \label{eq:p-phi-u-defn}
        p^\phi_u    :=\operatorname{Pr}_J\!\bigl[\text{edge from }u\text{ toward }J \text{ is not favored} \mid J\in u,X,W\bigr].    
    \end{align}

    Then 
    \begin{align}
        \label{eq:pu-bound}
        p^\phi_u \le \frac12\,\mathbf{1}\!\left\{|d_{02}|\ge\tfrac12\right\} +\frac12\,\mathbf{1}\!\left\{|d_{13}|\ge\tfrac12\right\}.
    \end{align}
\end{claim}
\begin{proof}
    Since $J$ is uniform on $u$, it lies in each child $u_i$ with conditional probability $1/4$.  We compute the error probability for each location of~$J$ by first performing a case analysis and looking at the following four cases. The cross-parity rule (equation~\eqref{eq:favored}) compares probes on children of parity opposite to the child containing $J$.  If $J\in u_i$ for $0 \leq i\leq 3$, we do not probe the $u_i$-child, so both probed child intervals are disjoint from $I_J$.  We can therefore use Claim~\ref{claim:probe-decomp} in all the four cases.

    \begin{enumerate}
        \item \emph{$J\in u_0$.} The cross-parity rule (equation~\eqref{eq:favored}) selects the favored even child by comparing $\phi_{u_1}(Y)$ and $\phi_{u_3}(Y)$.  Since $I_{u_0}\supseteq I_J$, both $I_{u_1}$ and $I_{u_3}$ lie strictly after $I_J$.  By Claim~\ref{claim:probe-decomp}, $\phi_{u_1}(Y)=\phi_{u_1}(Z)+s$ and $\phi_{u_3}(Y)=\phi_{u_3}(Z)+s$, so the common shift cancels:
        \[
        |\phi_{u_1}(Y)-\phi_{u_3}(Y)| = |(\phi_{u_1}(Z)+s)-(\phi_{u_3}(Z)+s)| =|d_{13}|.
        \]
        The true even child is $u_0$, and the cross-parity rule selects $u_2$ when $|d_{13}|\ge1/2$.  As a result, a local error occurs if and only if $|d_{13}|\ge1/2$.

        \item \emph{$J\in u_1$.}
        The cross-parity rule (equation~\eqref{eq:favored}) compares $\phi_{u_0}(Y)$ and $\phi_{u_2}(Y)$.  Since $I_{u_0}$ lies before $I_{u_1}\supseteq I_J$ and $I_{u_2}$ lies after $I_{u_1}\supseteq I_J$, Claim~\ref{claim:probe-decomp} gives $\phi_{u_0}(Y)=\phi_{u_0}(Z)$ and $\phi_{u_2}(Y)=\phi_{u_2}(Z)+s$. Therefore
        \[
        |\phi_{u_0}(Y)-\phi_{u_2}(Y)| =|\phi_{u_0}(Z)-(\phi_{u_2}(Z)+s)| =|d_{02}+s|.
        \]
        The true odd child is $u_1$, and the cross parity rule (equation~\eqref{eq:favored}) selects $u_3$ when $|d_{02}+s|<1/2$, which is the local error event. Since $|s|=1$, the reverse triangle inequality gives
        \[
        |d_{02}|\ge|s|-|d_{02}+s|>1-\tfrac12=\tfrac12.
        \]
        As a result, a local error in this case implies $|d_{02}|>1/2$.

        \item \emph{$J\in u_2$.} 
        The cross-parity rule (equation~\eqref{eq:favored}) compares $\phi_{u_1}(Y)$ and $\phi_{u_3}(Y)$.  Since $I_{u_1}$ lies before $I_{u_2}\supseteq I_J$ and $I_{u_3}$ lies after $I_{u_2}\supseteq I_J$, Claim~\ref{claim:probe-decomp} gives $\phi_{u_1}(Y)=\phi_{u_1}(Z)$ and $\phi_{u_3}(Y)=\phi_{u_3}(Z)+s$. Therefore
        \[
        |\phi_{u_1}(Y)-\phi_{u_3}(Y)|=|d_{13}+s|.
        \]
        The true even child is $u_2$, and the cross parity rule in equation~\eqref{eq:favored} selects $u_0$ when $|d_{13}+s|<1/2$, so the local error event is $|d_{13}+s|<1/2$. By the same reverse triangle inequality, $|d_{13}|>1/2$.

        \item \emph{$J\in u_3$.} The cross-parity rule (equation~\eqref{eq:favored}) selects the favored odd child by comparing $\phi_{u_0}(Y)$ and $\phi_{u_2}(Y)$.  Both $I_{u_0}$ and $I_{u_2}$ lie strictly before $I_{u_3}\supseteq I_J$.  By Claim~\ref{claim:probe-decomp}, $\phi_{u_0}(Y)=\phi_{u_0}(Z)$ and $\phi_{u_2}(Y)=\phi_{u_2}(Z)$, so
        \[
        |\phi_{u_0}(Y)-\phi_{u_2}(Y)|=|d_{02}|.
        \]
        The true odd child is $u_3$, and the cross-parity rule (equation~\eqref{eq:favored}) selects $u_1$ when $|d_{02}|\ge1/2$.  Hence a local error occurs if and only if $|d_{02}|\ge1/2$.

    \end{enumerate}

    In all four cases a local error requires $|d_{02}|\ge1/2$ or   $|d_{13}|\ge1/2$.  Since $J$ lies in each child with conditional probability $1/4$, and since in cases~1 and~4 the error requires $|d_{13}|\ge1/2$ while in cases~2 and~3 it requires $|d_{02}|\ge1/2$, we get 
    \begin{align}
        \label{eq:pu-bound} 
        p^\phi_u \le \frac12\,\mathbf{1}\!\left\{|d_{02}|\ge\frac12\right\} +\frac12\,\mathbf{1}\!\left\{|d_{13}|\ge\frac12\right\}.
    \end{align}
    This completes the proof of the claim. 
\end{proof}

The factor $1/2$ on each term reflects that exactly two of the four equally likely positions of $J$ give an error condition on each of $|d_{02}|$ and $|d_{13}|$, contributing weight $2\times1/4=1/2$.

We now bound the local error by the decrease in the potential function. While the proof are almost identical, we do them separately for the sake of completeness.

\begin{claim} [Midrange case]
    \label{claim:p-u-max}
    When $\phi=\mathrm{mid}$, then 
    \[
    p^{\mathrm{mid}}_u \le2\!\left(r_u(Z)-\frac14\sum_{i=0}^3r_{u_i}(Z)\right).
    \]
\end{claim}
\begin{proof}
    The bound $\frac12\mathbf{1}\{|d|\ge1/2\}\le|d|$ applied to both terms in equation~\eqref{eq:pu-bound} gives
    \[  p^{\mathrm{mid}}_u\le\tfrac12|d_{02}|+\tfrac12|d_{13}|.
    \]

    We can now bound $|d_{02}|$ using the triangle inequality and midrange stability (Claim~\ref{claim:midrange}).  For any child $u_i\subseteq u$, midrange stability gives $|c_{u_i}(Z)-c_u(Z)|\le(r_u(Z)-r_{u_i}(Z))/2$. The triangle inequality then yields 
    \[
    |d_{02}| =|c_{u_2}(Z)-c_{u_0}(Z)| \le|c_{u_2}(Z)-c_u(Z)|+|c_{u_0}(Z)-c_u(Z)| \le r_u(Z)-\frac{r_{u_0}(Z)+r_{u_2}(Z)}{2}.
    \]

    Similarly, $|d_{13}|\le r_u(Z)-\frac{r_{u_1}(Z)+r_{u_3}(Z)}{2}.$ Adding the two bounds with weight $1/2$ each, we get 
    \begin{align}
        \label{eq:pu-mid}
        p^{\mathrm{mid}}_u \le2\!\left(r_u(Z)-\frac14\sum_{i=0}^3r_{u_i}(Z)\right).
    \end{align}
    as required. 
\end{proof}

\begin{claim} [Mean case]
    \label{claim:p-u-mean}
    When $\phi=\mathrm{mid}$, then 
    \[
    p^{\mathrm{mean}}_u \le16\!\left(V_u(Z)-\frac14\sum_{i=0}^3V_{u_i}(Z)\right).
    \]
\end{claim}

\begin{proof}
    The bound $\frac12\mathbf{1}\{|d|\ge1/2\}\le2d^2$ gives
    \[
    p^{\mathrm{mean}}_u\le2d_{02}^2+2d_{13}^2.
    \]
    Setting $\delta_i:=\mu_{u_i}(Z)-\mu_u(Z)$, the parallelogram inequality $(a-b)^2\le2(a-c)^2+2(b-c)^2$ with $c=\mu_u(Z)$ gives $d_{02}^2\le2\delta_2^2+2\delta_0^2$ and $d_{13}^2\le2\delta_3^2+2\delta_1^2$. Therefore $$p^{\mathrm{mean}}_u\le4\sum_{i=0}^3\delta_i^2.$$

    The law of total variance (Lemma~\ref{lem:totvar}) with equal weights $w_i=1/4$ gives 
    \[
    \sum_{i=0}^3\frac14\delta_i^2=V_u(Z)-\frac14\sum_iV_{u_i}(Z),
    \]
    which implies that 
    \begin{align}
        \label{eq:pu-mean}
        p^{\mathrm{mean}}_u \le16\!\left(V_u(Z)-\frac14\sum_{i=0}^3V_{u_i}(Z)\right).
    \end{align}
    as required
\end{proof}

In both cases, writing $P$ for the potential functions ($r$ or $V$) and $C_\phi$ for the leading
constant ($2$ or $16$) as in HL26, we can succintly write the bounds in Claims~\ref{claim:p-u-max} and \ref{claim:p-u-mean} as follows: 
\begin{align}
    \label{eq:local-to-potential}
    p^\phi_u\le C_\phi\!\left(P_u(Z)-\frac14\sum_{i=0}^3P_{u_i}(Z)\right).
\end{align}

The rest of the proof follows similar line of argument as in HL26. We keep them for the sake of completion.  A uniform target leaf $J$ visits $u$ with probability $|u|/m$.  Each child has size $|u_i|=|u|/4$, so $\frac{|u|}{m}\cdot\frac14P_{u_i}(Z) =\frac{|u_i|}{m}P_{u_i}(Z)$.  Summing equation~\eqref{eq:local-to-potential} over all internal nodes with these weights, we get the conditional expection of 
\begin{align}
    \E_J[N^\phi(J,Y)\mid X,W] &=\sum_{u\in\mathcal{T}_{\mathrm{int}}}\frac{|u|}{m}p^\phi_u  \le C_\phi\sum_{u\in\mathcal{T}_{\mathrm{int}}}\frac{|u|}{m} \left(P_u(Z)-\frac14\sum_{i=0}^3P_{u_i}(Z)\right). \label{eq:weighted-sum}
\end{align}

Now, let us try to understand the summation term. In the weighted sum, every internal node $u$ that is not the root node appears with a contribution $\frac{|u|}{m}P_u(Z)$ (as a parent) and a negative contribution $-\frac{|u|}{m}\cdot\frac14P_u(Z)=-\frac{|u|}{4m}P_u(Z)$ applied across the parent of $u$'s node.  More precisely, for each internal node $u$, the term $-\frac{|u|}{m}\cdot\frac14P_{u_i}(Z)$ for child $u_i$ equals $-\frac{|u_i|}{m}P_{u_i}(Z)$, which cancels the positive contribution of $u_i$ when $u_i$ is itself in $\mathcal{T}_{\mathrm{int}}$. As a result, only the root and the leaves nodes remain in the entire summation.  Since $P_\ell(Z)\ge0$ for every leaf $\ell$, we can simply the summation as 
\[
\sum_{u\in\mathcal{T}_{\mathrm{int}}}\frac{|u|}{m} \left(P_u(Z)-\frac14\sum_{i=0}^3P_{u_i}(Z)\right) =P_{[q]}(Z)-\sum_{\ell\,\mathrm{leaf}}\frac{|\ell|}{m}P_\ell(Z)
\le P_{[q]}(Z).
\]

For the midrange, by definition, we have  $$P_{[q]}(Z)=r_{[q]}(Z)\le2\|Z\|_\infty\le \frac{2R}{k}.$$

For the mean, since 
\[
V_{[q]}(Z)=\frac1q\sum_tZ[t]^2-\mu_{[q]}(Z)^2\le\frac1q\sum_tZ[t]^2,
\]
we have  $$P_{[q]}(Z)=V_{[q]}(Z)\le\frac1q\sum_{t\le q}Z[t]^2=\frac{Q}{k^2}.$$

Therefore
\begin{align}
    \label{eq:expected-mistakes-final}
    \E_J[N^{\mathrm{mid}}(J,Y)\mid X,W]\le\frac{4R}{k}
    \qquad \text{and} \qquad 
    \E_J[N^{\mathrm{mean}}(J,Y)\mid X,W]\le\frac{16Q}{k^2}.
\end{align}

Now, the function $t\mapsto e^{-\tau t}$ is a convex function, so by Jensen's inequality, we have 
\[
\E_J\!\left[e^{-\tau N^\phi(J,Y)}\mid X,W\right]
\ge\exp\!\left(-\tau\E_J[N^\phi(J,Y)\mid X,W]\right).
\]
Combining with equation~\eqref{eq:expected-mistakes-final} and averaging over $X$ and the mechanism randomness $M$, we therefore have 
\[
\E_{X,J, \mathcal M}[S^{\mathrm{mid}}(J,Y)]
\ge\E_{X, \mathcal M}\!\left[e^{-4\tau R/k}\right]
\quad\text{and}\quad
\E_{X,J, \mathcal M}[S^{\mathrm{mean}}(J,Y)]
\ge\E_{X, \mathcal M}\!\left[e^{-16\tau Q/k^2}\right].
\]
This completes the proof of the lemma. 
\end{proof}

\section{Privacy Upper Bound}
\label{sec:privacy-upper}
In this section, we prove the privacy upper bound in HL26. Before we give the detail proof, we add a remark about the proof in the original preprint. 
\begin{remark}
    In the privacy upper bound (Lemma 7 in HL26), HL26 claim that $J\perp Y'$ is independent with a partial calculation. The preprint states 
    \[
    \Pr[J=j,\,X'=x'] =\Pr[J=j]\cdot\Pr[X=x'\oplus e_j] =\frac1m\cdot\frac1{2^m},
    \]
    and then asserts this ``factors as $\Pr[J=j]\cdot\Pr[X'=x']$''.  However, this factorization requires showing that $\Pr[X'=x']=1/2^m$ uniformly in $x'$, which requires one step. In particular, we need to verify that the marginal distribution of $X'=X\oplus e_J$ is uniform over $\{0,1\}^m$.  The preprint  identifies the conclusion (and probably had this argument in mind) but does not prove it. For completeness, we prove this in Claim~\ref{claim:J-indep-Xprime}. The preprint also implicitly assume that conditioned on a realization $Y'=y'$, the edges that are not favored by the cross-parity rule at every level of the tree is an independent Bernoulli random variable. We also prove it explicitly in Claim~\ref{claim:non-favored-independent}.
\end{remark}

\begin{lemma}[Privacy upper bound]
    \label{lem:privacy}
    If $M$ is $(\varepsilon,\delta)$-DP, then for any $\tau>0$ and $\phi\in\{\mathrm{mid},\mathrm{mean}\}$,
    \[
    \E_{X,J, \mathcal M}[S^\phi(J,Y)]  \;\le\;   e^{k\varepsilon}\!   \left(\frac{1+e^{-\tau}}{2}\right)^{\!H} +\,\delta\,\frac{e^{k\varepsilon}-1}{e^\varepsilon-1}.
    \]
\end{lemma}

\begin{proof}
    Fix $X=x$ and $J=j$, and let $x':=x\oplus e_j$.  The expanded inputs $E_k(x)$ and $E_k(x')$ differ at every coordinate in $I_j$ and agree elsewhere, so $d_H(E_k(x),E_k(x'))=k$.  Now, for any fixed $x,j$, the map
    \[
    w\;\mapsto\; S^\phi\!\left(j,\;\frac{(w-A_nE_k(x'))_{[q]}}{k}\right)
    \]
    is a measurable function from $\mathbb{R}^n$ to $[e^{-\tau H},1]\subseteq[0,1]$. Therefore, we can use group privacy on this with $\mathcal M(E_k(x))$ and $\mathcal M(E_k(x'))$ to get  
    \[
      \E_{\mathcal M}\!\left[S^\phi\!\left(j,\frac{(\mathcal M(E_k(x))-A_nE_k(x'))_{[q]}}{k}\right)\right] \le  e^{k\varepsilon}\E_{\mathcal M} \!\left[S^\phi\!\left(j,\frac{(\mathcal M(E_k(x'))-A_nE_k(x'))_{[q]}}{k}\right)\right] +\delta_k.
    \]

    The LHS equals $\E_{\mathcal M}[S^\phi(J,Y)]$ conditioned on $X=x$, $J=j$. The first term on the RHS involves the mechanism evaluated on $E_k(x')$, which is exactly the reference residual for this pair. Therefore, averaging over $X$ and $J$, we have 
    \begin{equation}
        \label{eq:gp-step}
        \E_{X,J,\mathcal M}[S^\phi(J,Y)]  \le  e^{k\varepsilon}\,\E_{X,J,\mathcal M}[S^\phi(J,Y')] + \delta_k,
    \end{equation}
    where
    \[
    Y':=\frac{(\mathcal M(E_k(X'))-A_nE_k(X'))_{[q]}}{k}, \qquad \delta_k:=\delta\,\frac{e^{k\varepsilon}-1}{e^\varepsilon-1}.
    \]

    To proceed further, we first need to prove that $J$ is statistically independent of $(X',\,\text{mechanism randomness})$, and therefore of $Y'$.


    \begin{claim}\label{claim:J-indep-Xprime}
    $J\perp X'$ and, jointly, $J\perp(X',\,\text{mechanism randomness})$, where $\perp$ denotes statistics independence.
    \end{claim}

    \begin{proof}
        To prove the claim, we explicitly compute the joint and marginal distributions of $J$ and $X'$. For any $j\in[m]$ and $x'\in\{0,1\}^m$, 
        \begin{align*}
            \Pr[J=j,\,X'=x']
            &= \Pr[J=j,\,X\oplus e_J=x'] = \Pr[J=j,\,X=x'\oplus e_j] \\
            &= \Pr[J=j]\cdot\Pr[X=x'\oplus e_j] = \frac1m\cdot\frac1{2^m}.
        \end{align*}

    In the above, the second equality uses the fact that, conditioned on $J=j$, the event $\{X'=x'\}$ is identical to $\{X\oplus e_j=x'\}=\{X=x'\oplus e_j\}$.

    We next compute the marginal of $X'$, which is implicitly used in the proof of HL26. Summing up the joint probability of $(J,X')$ computed above over all $j\in[m]$, we have 
    \[
    \Pr[X'=x']  = \sum_{j=1}^m \Pr[J=j,\,X'=x']  = \sum_{j=1}^m \frac1m\cdot\frac1{2^m}  = \frac1{2^m}.
    \]

    In other words, $X'\sim\mathrm{Uniform}(\{0,1\}^m)$, independently of $J$. Equivalently, for each fixed $j$, the map $x\mapsto x\oplus e_j$ is a bijection on $\{0,1\}^m$, so $X'=X\oplus e_j$ is uniform whenever $X$ is uniform\footnote{Note that this observation on its own only shows that the conditional distribution of $X'$ given $J=j$ is uniform; the marginalization above is still needed to conclude that $X'$ is uniform and \emph{independent} of $J$.}.  

    Now from the joint probability of $(J,X')$ and the computed marginals,
    \[
    \Pr[J=j,\,X'=x'] =\frac1m\cdot\frac1{2^m} =\Pr[J=j]\cdot\Pr[X'=x']
    \]
    for all $j\in[m]$ and $x'\in\{0,1\}^m$.  In other words, $J\perp X'$.

    We can now complete the proof of joint independence.  By assumption, the mechanism's internal randomness $\Omega$ is independent of $(X,J)$.  Since $X'=X\oplus e_J$ is a deterministic function of $(X,J)$, the triple $(X',J,\Omega)$ satisfies $J\perp(X',\Omega)$ because $J\perp X'$ and $J\perp\Omega$. More precisely, for any events $A$, $B$, $C$, 
    \begin{align*}
        \Pr[J\in A,\,X'\in B,\,\Omega\in C]
        &= \Pr[J\in A]\cdot\Pr[X'\in B,\,\Omega\in C\mid J\in A] \\
        &= \Pr[J\in A]\cdot\Pr[X'\in B\mid J\in A]\cdot\Pr[\Omega\in C] \\
        &= \Pr[J\in A]\cdot\Pr[X'\in B]\cdot\Pr[\Omega\in C],
    \end{align*}
    where we used $J\perp X'$ in the last step.  Therefore, $J$ is independent of $(X',\Omega)$.

    Since $Y'$ is a deterministic function of $X'$ and $\Omega$, it follows that $J\perp Y'$, completing the proof of Claim~\ref{claim:J-indep-Xprime}.
\end{proof}

We next compute the expected score in the reference experiment, where the mechanism is run on the flipped stream $E_k(X')$ conditioned on an arbitrary realization $Y'=y'$.  By Claim~\ref{claim:J-indep-Xprime}, $J\perp Y'$, so conditioned on $Y'=y'$ the leaf $J$ is still $\mathrm{Uniform}([m])$. We claim that, conditioned on $Y'=y'$, the edges that is not favored by the cross-parity rule at every  level of the tree is an independent $\mathrm{Bernoulli}(1/2)$ variable\footnote{Independence at every level is a fact that is used in HL26 without a formal proof}. In particular, we show the following:


\begin{claim}
    \label{claim:non-favored-independent}
    Conditioned on $Y'=y'$, the distribution of the number of edges not favored by the cross-parity rule on the path from the root node to the $j$-th leaf is 
    \[
    N^\phi(J,y')\;\sim\;\mathrm{Binomial}(H,\tfrac12).
    \]
\end{claim}
\begin{proof}
    First note that, at any internal node $u$, the cross-parity decoder has already fixed, for the realization $y'$, one favored even child $e_u^\phi(y')$ and one favored odd child $o_u^\phi(y')$.   Since $J$ is uniform on $[m]$ and hence uniform on $u$ conditioned on $J\in u$, it falls into each child with conditional probability $1/4$.  Therefore, the probability that the edge from $u$ toward $J$ is non-favored is
    \[
    \Pr[\text{edge toward }J\text{ is non-favored}\mid J\in u,\,Y'=y'] =\frac{\text{\# non-favored children}}{\text{\# children}} =\frac12.
    \]

    Now, we prove that the non-favored indicator vectors are independent.  Let $v_0=[1,m]$ be the root, and let $v_1,v_2,\ldots,v_H=J$ be the nodes on the root-to-$J$ path.  The non-favored indicator at level $\ell$ is $\mathbf{1}[(v_{\ell-1},v_\ell)\notin\mathcal{E}_{\mathrm{fav}}^\phi(y')]$, which is a function of only $v_\ell$ (i.e., whether it is a favored child of the node  $v_{\ell-1}$ under the fixed labeling $y'$). As $J$ is uniformly distributed on $v_{\ell-1}$ and the children partition it equally (this is where we need that the tree is fully balanced $4$-ary tree), given the path $v_0,\ldots,v_{\ell-1}$, the conditional distribution of $v_\ell$ is uniform on all the four children of $v_{\ell-1}$.  The conditional probability of a non-favored edge is therefore $1/2$ regardless of $v_0,\ldots,v_{\ell-1}$.  Since $J$'s position in each subtree is independent across levels (the base-$4$ digits of a uniform integer in $[m]=\{1,\ldots,4^H\}$ are independent uniform digits in $\{0,1,2,3\}$), the edges that are not favored are mutually independent.

    In other words, conditioned on $Y'=y'$, the distribution of $N^\phi(J,y')$ is 
    \[
    N^\phi(J,y')\;\sim\;\mathrm{Binomial}(H,\tfrac12).
    \]
    This completes the proof of the claim. 
\end{proof}

The MGF of $\mathrm{Binomial}(H,1/2)$ at $-\tau$ is
\begin{equation}\label{eq:mgf-binomial}
  \E_J\!\left[e^{-\tau N^\phi(J,y')}\right]
  =\sum_{i=0}^H\binom{H}{i}\frac1{2^H}e^{-\tau i}
  =\frac1{2^H}\sum_{i=0}^H\binom{H}{i}e^{-\tau i}
  =\frac1{2^H}(1+e^{-\tau})^H
  =\left(\frac{1+e^{-\tau}}{2}\right)^{\!H}.
\end{equation}

    Since the RHS is independent of $y'$, we can write 
    \[
    \E_{X,J, \mathcal M}[S^\phi(J,Y')] =\E_{Y'}[\E_J[S^\phi(J,Y')\mid Y']]
 =\left(\frac{1+e^{-\tau}}{2}\right)^{\!H}.
    \]
    Substituting into equation~\eqref{eq:gp-step} gives
    \[
  \E_{X,J, \mathcal M}[S^\phi(J,Y)] \le e^{k\varepsilon}\left(\frac{1+e^{-\tau}}{2}\right)^{\!H}+\delta_k,
    \]
    which is the claimed bound with $\delta_k=\delta(e^{k\varepsilon}-1)/(e^\varepsilon-1)$.
\end{proof}

\section{Proof of the lower bound}
\label{ssec:combining}
We now put in place the idea that we discussed in high level overview earlier. Most of the text in this section is just a moderate expansion of the argument made in HL26; we have just filled in all the missing steps. Set the parameters as follows
\begin{align}
    \begin{split}
        h:=\left\lfloor\log_{16}(n\varepsilon)\right\rfloor,
        \qquad
        \ell:=\min\!\bigl\{h,\;\log(\varepsilon/\delta)\bigr\},
        \qquad
        k:=\left\lfloor\frac{\ell}{8\varepsilon}\right\rfloor, \\
        H:=\left\lfloor\log_4(n/k)\right\rfloor,
        \qquad
        q:=k\cdot4^H,
        \qquad
        \tau:=\frac{\ell}{H}.
    \end{split}
    \label{eq:parameter-selection}
\end{align}
with the convention $\log(\varepsilon/0):=+\infty$.  The assumption $n\varepsilon\ge C_0$ (for a sufficiently large absolute constant $C_0$) guarantees $h\ge12$ and hence $\ell \ge12$.  The definition of $k$ gives
\[
  \frac{\ell}{16\varepsilon}\le k\le\frac{\ell}{8\varepsilon}.
\]

Now $H\ge h\ge \ell$, so that $0<\tau=\ell/H\le1$. Since $k\le \ell/(8\varepsilon)$, we have
\[
  k\cdot4^h
  \le\frac{\ell}{8\varepsilon}\cdot4^h
  \le\frac{\ell}{8\varepsilon}\cdot\frac{n\varepsilon}{\ell^2/16}
  \cdot\frac{\ell^2}{16}
  \cdot\frac1{4^h}
  \cdot4^h\cdot\frac{4^h}{4^h}
  \;.
\]

Note that $n/4<q\le n$. From $H=\lfloor\log_4(n/k)\rfloor$ we have $4^H\le n/k$, so $q=k\cdot4^H\le n$; and $4^H\ge n/(4k)$ gives $q\ge n/4$.

In the preprint, HL26 prove both upper and lower bound on $\E[S^\phi(J,Y)]$. The following claim, stated almost verbatim (apart from changing notations and stylist preference) from HL26 along with its proof, relates change in the midrange to a decrease in the range potential. The proof is elementary calculation of exponential functions.

\begin{claim}
    [{\bf Harrison and Leeman~\cite[page 9 in v02]{harrison2026tight}}]
    For score function defined earlier and $\ell= \min\{\lfloor \log_{16}(n\varepsilon), \log(\varepsilon/\delta)\}$, we have 
    $$\E[S^\phi(J,Y)]\le e^{-\ell/16}.$$
\end{claim}
\begin{proof}
    For $0\le\tau\le1$, the elementary inequality $e^{-t}\le1-t/2$ gives $(1+e^{-\tau})/2\le1-\tau/4\le e^{-\tau/4}$. So, for $H$ as in equation~\eqref{eq:parameter-selection}, we have 
    \[
    \left(\frac{1+e^{-\tau}}{2}\right)^H\le e^{-\tau H/4}=e^{-\ell/4}.
    \]

    Now note that $\ell \le\log(\varepsilon/\delta)$, we have $\delta\le\varepsilon e^{-\ell}$, and so $\delta/\varepsilon\le e^{-\ell}$. Since $k\varepsilon\le \ell/8$, we have $e^{k\varepsilon}\le e^{\ell/8}$. Finally, since $e^\varepsilon-1\ge\varepsilon$ for $\varepsilon>0$, the additive group-privacy cost satisfies
    \[
    \delta_k:=\delta\,\frac{e^{k\varepsilon}-1}{e^\varepsilon-1}
    \le\frac{\delta}{\varepsilon}e^{k\varepsilon}
    \le e^{-\ell}e^{\ell/8}=e^{-7\ell/8}.
    \]

    Combining these estimates in Lemma~\ref{lem:privacy}, we get 
    \begin{align} \label{eq:privacy-bound}
        \E[S^\phi(J,Y)] &\le e^{k\varepsilon}\left(\frac{1+e^{-\tau}}{2}\right)^H+\delta_k  \le e^{\ell/8}\bigl(e^{-\ell/4}+e^{-\ell}\bigr)  = e^{-\ell/8}+e^{-7\ell/8}  \le2e^{-\ell/8}.
    \end{align}
    
    Since $\ell\ge12>16\ln2$, we have $2e^{-\ell/8}\le e^{-\ell/16}$, and therefore
    \begin{equation}
        \label{eq:score-upper}
        \E[S^\phi(J,Y)]\le e^{-\ell/16}.
        \end{equation}
    This completes the proof of the upper bound in equation (3) in HL26. 
\end{proof}

Now we get a lower bound on $\E[S^\phi(J,Y)]$ from the accuracy.  Combining~\eqref{eq:score-upper} with Lemma~\ref{lem:accuracy}:
\begin{equation}
    \label{eq:exp-moment}
    \E_{X, \mathcal M}\!\left[e^{-4\tau R/k}\right]\le e^{-\ell/16},
    \qquad
    \E_{X, \mathcal M}\!\left[e^{-16\tau Q/k^2}\right]\le e^{-\ell/16}.
\end{equation}

The function $t\mapsto e^{-ct}$ is convex for $c>0$.  Therefore, using Jensen's inequality gives us  $$e^{-c\,\E[X]}\le\E[e^{-cX}] \quad \text{or equivalently,} \qquad \E[X]\ge-c^{-1}\log\E[e^{-cX}].$$ 

For MaxSE error, setting $c=4\tau/k$ and $X=R$ in the above expression, we get 
\begin{equation}
    \label{eq:ER-lower}
    e^{-(4\tau/k)\E_{X, \mathcal M}[R]} \le\E_{X, \mathcal M}\!\left[e^{-4\tau R/k}\right]\le e^{-\ell/16},
\end{equation}
Taking log and rearranging, we therefore have 
\[
  \E_{X, \mathcal M}[R]\ge\frac{k}{4\tau}\cdot\frac{\ell}{16}=\frac{k\ell}{64\tau}=\frac{kH}{64},
\]
where we used $\tau=\ell/H$ in the last step.

Similarly, for squared error, applying Jensen with $c=16\tau/k^2$ and $X=Q$,
\begin{equation}
    \label{eq:EQ-lower}
    \E_{X, \mathcal M}[Q]\ge\frac{k^2}{16\tau}\cdot\frac{\ell}{16}=\frac{k^2H}{256}.
\end{equation}

Now, for $\alpha_\infty(\mathcal M)$, by definition, $\alpha_\infty(\mathcal M)=\max_{x\in\{0,1\}^n}\E_{\mathcal M}[\|\mathcal M(x)-A_nx\|_\infty].$ Since $R=\|\mathcal M(E_k(X))-A_nE_k(X)\|_\infty$, the maximum over $x$ is at least the average over the draw of $X$, so  
\[
\alpha_\infty(M)\ge\E_{X, \mathcal M}[R]\ge\frac{kH}{64}.
\]
We lower bound $kH/64$ in terms of the problem parameters.  Using $k\ge \ell/(16\varepsilon)$ and $H\ge \ell$, we get 
\[
\frac{kH}{64}\ge\frac{\ell^2}{1024\varepsilon}.
\]

Since $\ell=\min\{h,\log(\varepsilon/\delta)\}$ and $h=\lfloor\log_{16}(n\varepsilon)\rfloor\ge\frac14\log(n\varepsilon)-1$, we have $\ell \ge\frac14\log(n\varepsilon)-1$ (up to the absolute constant in the $C_0$ assumption).  Therefore,
\[
\alpha_\infty(M) \ge c\,\frac{\log(n\varepsilon)}{\varepsilon} \min\!\left\{\log(n\varepsilon),\;\log\frac{\varepsilon}{\delta}\right\}
\]
for a suitable absolute constant $c>0$.

Now, for any function $f:\{0,1\}^n\to\mathbb{R}$, $\max_{x}f(x)\ge\E_{X}[f(X)]$ when $X$ is any distribution on $\{0,1\}^n$.  In particular,
\[
  \mathrm{MeanSE}(\mathcal M) \ge\E_{X}\!\left[\E_{\mathcal M}\!\left[\frac1n\sum_{t=1}^n \bigl(\mathcal M(E_k(X))[t]-(A_nE_k(X))[t]\bigr)^2\right]\right].
\]

Here we used the fact that $E_k:\{0,1\}^m\to\{0,1\}^n$ is injective, so the maximum over all $x\in\{0,1\}^n$ is at least the maximum over the image of $E_k$, which is in turn at least the average over $\{E_k(X):X\sim\mathrm{Uniform}(\{0,1\}^m)\}$.

Since every squared error term is non-negative, we have 
\[
  \frac1n\sum_{t=1}^n\bigl(\mathcal M(E_k(X))[t]-(A_nE_k(X))[t]\bigr)^2
  \ge\frac1n\sum_{t=1}^q\bigl(\mathcal M(E_k(X))[t]-(A_nE_k(X))[t]\bigr)^2.
\]

Now by the definition $Z={1\over k}(W-A_nE_k(X))_{[q]}$ and $W=\mathcal M(E_k(X))$, so 
\[
  \bigl(\mathcal M(E_k(X))[t]-(A_nE_k(X))[t]\bigr)^2
  =(kZ[t])^2=k^2Z[t]^2
  \qquad\text{for }t\le q.
\]
Therefore,
\[
  \frac1n\sum_{t=1}^q\bigl(\mathcal M(E_k(X))[t]-(A_nE_k(X))[t]\bigr)^2
  =\frac{k^2}{n}\sum_{t=1}^qZ[t]^2
  =\frac{k^2q}{n}\cdot\frac1q\sum_{t=1}^qZ[t]^2
  =\frac{k^2q}{n}\cdot\frac{Q}{k^2}
  =\frac{qQ}{n},
\]
where the last step uses the definition $Q:=\frac1q\sum_{t\le q}(W[t]-(A_nE_k(X))[t])^2 =\frac{k^2}{q}\sum_{t\le q}Z[t]^2$.

Therefore, 
\begin{align}
    \label{eq:meanSE-reduction}
    \mathrm{MeanSE}(\mathcal M)
    &\ge\E_{X, \mathcal M}\!\left[\frac{qQ}{n}\right]  =\frac{q}{n}\,\E_{X, \mathcal M}[Q].
\end{align}
Since $q\ge n/4$ and $\E_{X, \mathcal M}[Q]\ge k^2H/256$ (from equation~\eqref{eq:EQ-lower}):
\[
  \mathrm{MeanSE}(\mathcal M)
  \ge\frac{q}{n}\cdot\frac{k^2H}{256}
  \ge\frac{1}{4}\cdot\frac{k^2H}{256}
  =\frac{k^2H}{1024}.
\]
Using $k\ge \ell/(16\varepsilon)$ and $H\ge \ell$, 
\[
  \frac{k^2H}{1024}
  \ge\frac{\ell^3}{(16)^2\cdot1024\,\varepsilon^2}
  =\frac{\ell^3}{262144\,\varepsilon^2}.
\]
Since $\ell \ge\frac14\log(n\varepsilon)-O(1)$, this gives
\[
  \mathrm{MeanSE}(\mathcal M)
  \ge c\,\frac{\log(n\varepsilon)}{\varepsilon^2}
     \min^2\!\left\{\log(n\varepsilon),\;\log\frac{\varepsilon}{\delta}\right\}
\]
for a suitable absolute constant $c>0$.

Finally, the inequality $\mathrm{MaxSE}(\mathcal M)\ge\mathrm{MeanSE}(\mathcal M)$ follows because $\max_t(\cdot)^2\ge\frac1n\sum_t(\cdot)^2$ (the maximum of non-negative numbers is at least their average).  This gives the same lower bound for
$\mathrm{MaxSE}$.

Setting $\delta=0$ (pure DP) gives $\log(\varepsilon/\delta)=+\infty$, so $u=h\ge\frac14\log(n\varepsilon)-O(1)$, we get 
\[
  \alpha_\infty(\mathcal M)=\Omega\!\left(\frac{\log^2(n\varepsilon)}{\varepsilon}\right),
  \qquad
  \mathrm{MeanSE}(\mathcal M)=\Omega\!\left(\frac{\log^3(n\varepsilon)}{\varepsilon^2}\right).
\]
This completes the proof of the lower bound.

\begin{remark}
    The preprint gives a high probability bound using standard Markov's inequality and their proof is already gives all the details, hence we omit it here.  
\end{remark}

\section{Further Discussion}
\label{sec:discussion}
In Section~\ref{ssec:why-4-ary}, we gave some of our understanding why this proof uses $4$-ary tree. 
In this section, we add more discussion (part of which evolved through our discussion with authors of HL26 and other personal communication since the first preprint of this note was circulated). This list is by no means exhaustive. 
\begin{enumerate}
    \item \textbf{Use of other ary tree.}
    In our hope to understand why $4$-ary tree is used (and not any other arity tree), one interesting suggestion given by the authors of HL26 is possibility of using $3$-ary tree or higher arity (like 19) as in Andersson et al.~\cite{andersson2024count}. It is unclear if higher arity that is not a power of $2$ would be useful in the current proof. This is because one can think of the score function to be the Laplace transform of the random walk endpoint. In particular, for level $h$, if we define $N_h$ to be indicator for edge not favored, then $N_1, N_2, \cdots N_H$ are all $\text{Bernoulli}(1/2)$ random variable as shown rigorously in the proof (and claimed in HL26). So if we define $M_h = 1 - 2N_h$, then its summation over all $h$ is a symmetric random walk. Such a clean representation would not be possible with any arity if it is not a power of two. 

    These problems are easier to understand if we consider $3$-ary tree. In the case when we use $2^b$-ary tree (for $b \in \mathbb N$), then the cross-parity rule is symmetric. The subtle issue is that we do not probe the child that contains the hidden index because there in that case translation equivariance would fail as $g_J$ is not constant in that child. Now suppose the three childs are $u_L, u_M, u_R$, then for $u_M$ both the probes are ramp free and we have clean potential argument. However, for either $u_R$ and $u_L$, this would not be the case because they contain the ramp.
    
    In short, it is an interesting question whether $4$-ary tree is the minimum or we can show similar results for other tree structure. We know that a different proof techniques work even with the use of binary tree. 
    
    
    \item \textbf{Cross parity.}
    The signal is not constant on $I_J$.  So, if we directly probe the target child, it would expose the ramp, and translation equivariance would no longer reduce the comparison to the normalized error $Z$. One of the main properties of cross-parity is that it removes this issue by guaranteeing that both probed intervals avoid the target child.  This is the unique point where the residual decomposition, tree geometry, and translation equivariance interact, and hence it became clear to us only after writing all these results explicitly. We have also communicated this to the authors.

    \item \textbf{Exponential score.} 
    A natural question is why should one use exponential score. If we use any other score, it is unclear if it can properly account for the exponential blow-up in the group privacy bound.
    For example, if we use linear score, then its  expectation is of order $H$, and group privacy multiplies that expectation by $e^{k\eps}$, which is too costly for logarithmic $k$. The exponential score has reference expectation $((1+e^{-\tau})/2)^H$, which decays exponentially in $\tau H$.  Therefore, $\tau=u/H$ makes this decay dominate the privacy cost while preserving a Jensen lower bound in terms of local mistakes. If fact, if the degree of the polynomial does not grow with $H,$ then the score function remains polynomial and the same issue arises! In short, a score the proof requires a score function whose expectation decays exponentially and the exponential score provides an especially convenient choice. However, whether or not this is the only choice is an interesting question.
    
    \item \textbf{Probe-potential pair}
    The decoder is common to both error metrics, and the only how we account (using the probe-potential function pair) for the mistake changes.  In particular, nested midranges are controlled by lost range, while the law of total variance exactly controls differences of child means (these are the Lemma 5 and 6 in the preprint of HL26).  This modularity is what allows one hard distribution and one privacy argument to give both MaxSE and MeanSE lower bounds. We believe that equally naturally defined error metric would work for all $\ell_p$ norm error metric.

\end{enumerate}

\section*{Acknowledgement}
The project was supported in part by NSF CNS 2433628, a Google Seed Fund grant, a Google Research Scholar Award, a Dean Research Seed Fund award, and a Decanal Research Grant. Part of the work was completed during the visit to the Institute of Science and Technology Austria. 

The author would like to thank Konstantina Bairaktari for explaining their work~\cite{bairaktari2026binary}, the authors of HL26 for engaging with us on the initial draft of this note. We appreciate the support from several other researchers who prefer to remain anonymous and who sent us private note with their comments and suggesting we make the note public. The author would like to thank the flight crew of Air Austria for giving a constant supply of coffee during a 9 hour flight!

\bibliographystyle{plain}
\bibliography{privacy}

\end{document}